\documentclass[conference]{IEEEtran}
\IEEEoverridecommandlockouts

\usepackage{cite}
\usepackage{algorithmic}
\usepackage{graphicx}
\usepackage{textcomp}
\usepackage{xcolor}

\newcommand{\bs}[1]{\boldsymbol{#1}} 
\newcommand{\mcl}[1]{\mathcal{#1}}

\makeatletter   
\newcommand{\vast}{\bBigg@{3}}
\newcommand{\Vast}{\bBigg@{4}}
\makeatother

\newcommand{\msf}[1]{\mathsf{#1}}

\usepackage[utf8]{inputenc}
\usepackage[T1]{fontenc}
\usepackage{url}
\usepackage{tikz}
\usepackage{pgfplots} 
\usepgfplotslibrary{groupplots} 
\usepgfplotslibrary{dateplot}   
\pgfplotsset{compat=newest}
\usepackage{pgf}  
\usepackage{ifthen}  
\usepackage{cite}
\usepackage[cmex10]{amsmath} 
\usepackage{amsfonts, amssymb} 
\usepackage[version=4]{mhchem}
\usepackage{graphicx, amsthm, float, dsfont,color}
\usepackage[utf8]{inputenc} 
\usepackage[T1]{fontenc}
\usepackage{url}
\usepackage{ifthen}
\usepackage{cite}

\newtheorem{theorem}{Theorem} 

\newtheorem{corollary}{Corollary}  

\newtheorem{proposition}{Proposition}

\newtheorem{remark}{Remark}

\begin{document}

\title{Universal Decoding for Unequal Error Protection}

\author{%
  \IEEEauthorblockN{Seyed AmirPouya Moeini}
  \IEEEauthorblockA{University of Cambridge\\
    \texttt{sam297@cam.ac.uk}}
  \and
  \IEEEauthorblockN{Albert Guill\'en i F\`abregas}
  \IEEEauthorblockA{University of Cambridge\\
  Universitat Polit\`ecnica de Catalunya\\
    \texttt{guillen@ieee.org}}
  \thanks{This research was supported in part by the European Research Council under Grants 101142747 and 101158232, the Spanish Government under Grants PID2020-116683GB-C22, PID2021-128373OB-I00 and PID2025-167766OA-C22.}
}

\maketitle

\begin{abstract}
	We study universal unequal error protection of a single special message over discrete memoryless channels.
	We construct a sequence of codes with a decoding rule that does not depend on the channel transition law.
	We derive an achievable error exponent for the special message while ensuring reliable decoding of the ordinary messages.
	The resulting exponent strictly improves on the natural equal-error-protection random-coding baseline and remains positive as the rate approaches capacity.
	For symmetric channels, it converges to the optimal red-alert exponent at capacity.
\end{abstract}


\section{Introduction}

Borade et al.~\cite{BNZ2009} studied the \emph{message-wise unequal error protection} (UEP) problem, in which a single message is designated as special and requires stronger protection than the remaining messages. 
This setting arises naturally in wireless networks, interactive communication, and control systems.

Let $W:\mcl{X}\to\mcl{Y}$ be a discrete memoryless channel (DMC) with finite input and output alphabets.
Consider a length-$n$ block code with message set $\mcl{M}_n=\{1,2,\ldots,\mathrm{M}_n\}$ where $\mathrm{M}_n=e^{nR}$ and $R>0$ is the transmission rate.
The messages are assumed to be equiprobable. 
Without loss of generality, the first message, $M=1$, is designated as the special message, while the messages in $\{2,\ldots,\mathrm{M}_n\}$ are referred to as ordinary messages.
Given the channel output, the decoder produces an estimate $\hat M$ of the transmitted message.


We require the average error probability of the ordinary messages to vanish,
\begin{equation}
  {p}_{\mathrm o}
 \triangleq \frac{1}{\mathrm{M}_n-1}\sum_{m=2}^{\mathrm{M}_n}
 \mathbb{P}\big[\hat M\neq m\,|\, M=m\big]
 \longrightarrow 0,
 \label{eq:ordinary-criterion}
\end{equation}
and we seek a large error exponent for the special message,
\begin{equation}
	 E_{\mathrm s}\triangleq
	 \liminf_{n\to\infty}-\frac1n 
	 \log\mathbb{P}\big[\hat M\neq1\,|\, M=1\big].
	 \label{eq:special-exponent}
\end{equation}

For rates approaching capacity $C$, Borade et al.~\cite{BNZ2009} showed that the optimal special-message exponent is the \emph{red-alert} exponent
\begin{align}
	\mcl{J}^{\star}
	\triangleq
	\max_{x\in\mcl{X}}\,
	D\!\left(P_Y^\star\middle\|W(\cdot|x)\right),
\end{align}
where $P_Y^\star$ is the capacity-achieving output distribution. 
At rates strictly below capacity, a larger exponent is achievable~\cite[Ch.~5]{boradethesis} (see also~\cite{WCCW2012}).
More precisely, the optimal achievable special-message exponent at rate $R$ is
\begin{align}
	\!\!\!\!\mcl{J}(R) \triangleq \!\!\!\!\!\! 
	\sup_{\substack{
		\alpha\in\mcl{P}(\mcl{X}),\,
		\{P_X^x\}_{x\in\mcl{X}}:\\
		\sum_{x}\!\alpha(x)I(P_X^xW)>R}}
	\sum_{x\in\mcl{X}}
	\alpha(x)
	D\!\left( (P_X^xW)_Y \| W(\cdot|x) \right),
	\label{eq:matched-exponent}
\end{align}
{where $(P_X^xW)_Y$ denotes the output distribution induced by $P_X^xW$.}
The corresponding problem for the additive white Gaussian noise channel was studied in~\cite{NazerShkelDraper}.  
All of these constructions require knowledge of the channel.
Such knowledge may be unavailable or difficult to obtain accurately in many applications. 
A universal decoder does not require knowledge of the channel transition law, and uses the same decoding rule for all DMCs with given input and output alphabets.
The rates and error exponents achievable with universal decoding have been studied extensively~\cite{csiszarkorner,FederLapidoth1998}.

In this work, we study unequal error protection when neither the encoder nor the decoder knows the channel.
To the best of our knowledge, this problem has not previously been studied in the universal setting.
We propose a code construction together with a decoding rule that does not depend on the channel transition probabilities.
We derive an achievable error exponent for the special message that is strictly larger than the random-coding exponent under equal error protection.
This exponent remains positive as the rate approaches capacity, while the average error probability of the ordinary messages vanishes.

%

\emph{Notation:} 
Scalar random variables are denoted by uppercase letters, their sample values by lowercase letters, and their alphabets by calligraphic letters. Random vectors are represented in boldface.
For positive sequences $f_n$ and $g_n$, we write $f_n\doteq g_n$ if $\lim_{n\to\infty}\frac1n\log(f_n/g_n)=0$, and $f_n\,\dot\leq\, g_n$ with the analogous $\limsup$ inequality. 
The set of probability distributions on a finite set $\mathcal{A}$ is denoted by $\mathcal{P}(\mathcal{A})$.
The type of a sequence $\bs{x}$ is its empirical distribution, denoted by $\hat{P}_{\bs x}$, and the joint type $\hat{P}_{\bs x\bs y}$ is defined analogously. 
For $P \in \mcl{P}(\mcl{X})$, the type class $\mcl{T}^n(P)$ consists of all sequences in $\mcl{X}^n$ with type $P$.
For a joint distribution $P$, we write $I(P)$ for its mutual information, and $[t]^+\triangleq\max\{t,0\}$.

\section{Code Construction}

For DMCs, the constant-composition random-coding exponent is universally
achievable using maximum mutual information (MMI) decoding~\cite[Ch.~10]{csiszarkorner} (see also~\cite{moeini2025-1}). 
A simple approach is therefore to use a constant-composition code with MMI decoding and designate one of its messages as special.
This guarantees the random-coding exponent for the special message, but provides no additional protection compared with the ordinary messages. Moreover, this exponent vanishes as the rate approaches capacity.

We seek a construction that allows the special message to be
distinguished from the ordinary messages without knowing the channel.
The idea is to divide every codeword into subblocks.
For the ordinary messages, all subblocks have the same composition.
For the special message, each subblock consists of repetitions of a different input symbol.

Without loss of generality, let $\mcl{X}=\{1,\ldots,K\}$.
Fix $Q,\alpha\in\mcl{P}(\mcl{X})$.
The distribution $\alpha$ determines the subblock lengths, while $Q$ specifies the composition of each ordinary subblock.
We use the same partition for every codeword: the subblock indexed by $x\in\mcl{X}$ has length $n_x=n\alpha(x)$.
Subblocks with $\alpha(x)=0$ are omitted.
For simplicity, we assume that $n_x$ and $n_xQ(a)$ are integers for all $x,a\in\mcl{X}$. Standard type approximation removes this restriction without affecting the error exponents.

\emph{Special message:}
In subblock $x$, we repeat the symbol $x$ exactly $n_x$ times.
The special codeword is therefore
\begin{equation}
  \bs x_1=
  \big(
    \underbrace{1,...,1}_{n_1},
    \ldots,
    \underbrace{K,...,K}_{n_K}
  \big).
  \label{eq:special-codeword}
\end{equation}

\emph{Ordinary messages:}
For each ordinary message, subblock $x\in\mcl{X}$ is drawn uniformly
from the type class $\mcl{T}^{n_x}(Q)$, independently across subblocks
and messages.
Thus, the ordinary codewords are drawn independently and uniformly from
\begin{equation}
  \mcl{A}_n(Q)
  \triangleq
  \mcl{T}^{n_{1}}(Q)
  \times \cdots \times
  \mcl{T}^{n_{K}}(Q).
  \label{eq:ordinary-ensemble}
\end{equation}
Every ordinary codeword therefore has composition $Q$ within each
subblock and, consequently, composition $Q$ overall.
Figure~\ref{fig:code-construction} illustrates the construction.

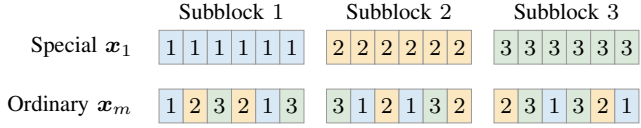
\begin{figure}[t]
\centering
\resizebox{\columnwidth}{!}{%
\begingroup
\definecolor{symbolone}{RGB}{216,232,245}
\definecolor{symboltwo}{RGB}{251,232,193}
\definecolor{symbolthree}{RGB}{218,233,217}
\newcommand{\drawsymbols}[3]{%
  \foreach[count=\i from 0] \s in {#3}{%
    \ifcase\s
      \or\def\cellcolor{symbolone}%
      \or\def\cellcolor{symboltwo}%
      \or\def\cellcolor{symbolthree}\fi
    \pgfmathsetmacro{\cellx}{#1+0.30*\i}
    \draw[fill=\cellcolor,draw=black!40,line width=0.25pt]
      (\cellx,#2) rectangle ++(0.30,0.42);
    \node at (\cellx+0.15,#2+0.21) {$\s$};
  }%
}
\begin{tikzpicture}[x=1cm,y=1cm,
  font=\fontsize{8}{9.5}\selectfont]
  \path[use as bounding box] (0,0) rectangle (7.70,1.66);

  \foreach \cx/\s in {2.55/1,4.63/2,6.71/3}
    \node at (\cx,1.48) {Subblock $\s$};

  \node[anchor=east] at (1.45,1.04) {Special $\bs x_1$};
  \drawsymbols{1.65}{0.83}{1,1,1,1,1,1}
  \drawsymbols{3.73}{0.83}{2,2,2,2,2,2}
  \drawsymbols{5.81}{0.83}{3,3,3,3,3,3}

  \node[anchor=east] at (1.45,0.31) {Ordinary $\bs x_m$};
  \drawsymbols{1.65}{0.10}{1,2,3,2,1,3}
  \drawsymbols{3.73}{0.10}{3,1,2,1,3,2}
  \drawsymbols{5.81}{0.10}{2,3,1,3,2,1}
\end{tikzpicture}
\endgroup
}
\caption{Example with $n=18$, $K=3$, and
$\alpha=Q=(1/3,1/3,1/3)$.
Each special subblock repeats one symbol, while each ordinary subblock
contains two copies of every symbol.}
\label{fig:code-construction}
\end{figure}

To see why this is useful, consider the output type in each subblock.
When an ordinary message is transmitted, every subblock has input
composition $Q$, so its output type concentrates around the same
distribution $P_Y=(QW)_Y$.
When the special message is transmitted, the output type in subblock
$x$ instead concentrates around $W(\cdot| x)$, which may differ
across subblocks.
Thus, the decoder can test whether the subblock output types are close
to one another without knowing the channel.
This motivates the decoding rule.

For $\bs y\in\mcl{Y}^n$, let $\hat{P}_{\bs y}^x$ denote the output type in block $x$. Choose $\delta_n\downarrow0$ such that $n\delta_n^2\to\infty$, and define
\begin{equation}
	 \mcl{G}_n\triangleq\left\{\bs y:  \max_{x\neq x'}\big\|\hat{P}_{\bs y}^x-\hat{P}_{\bs y}^{x'}\big\|_\infty>\delta_n \right\},  \label{eq:Gn}
\end{equation}
where the maximum is taken over the nonempty subblocks.
For a fixed threshold $\Delta\ge0$, the decoder declares the special message, $\hat M=1$, whenever
\begin{equation}
 \bs y\in\mcl{G}_n
 \quad\text{or}\quad
 \max_{m\ge2} \, I(\hat{P}_{\bs x_m\bs y})<\Delta.
 \label{eq:decoder}
\end{equation}
Otherwise, it applies MMI decoding among the ordinary messages, selecting the message whose codeword has the largest empirical mutual information with the received sequence.
The second condition provides a rejection rule for MMI decoding: if all ordinary codewords have empirical mutual information below $\Delta$, the decoder declares the special message.


\section{Main Results}

For $\alpha\in\mcl{P}(\mcl{X})$, define
\begin{equation}
  \widehat{\mcl{J}}_\alpha
  \triangleq
  \min_{P\in\mcl{P}(\mcl{Y})}
  \sum_{x\in\mcl{X}} \alpha(x)D(P\|W(\cdot|x)),
\end{equation}
and let
\begin{equation}
  \widehat{\mcl{J}}
  \triangleq
  \max_{\alpha\in\mcl{P}(\mcl{X})}
  \widehat{\mcl{J}}_\alpha.
  \label{eq:hateopt}
\end{equation}
Since the objective is linear in $\alpha$ and convex in $P$, Sion's minimax theorem gives
\begin{align}
	\widehat{\mcl{J}}
	&=  \min_{P \in \mcl{P}(\mcl{Y})} \max_{\alpha \in \mcl{P}(\mcl{X})} \, \sum_{x\in \mcl{X}} \alpha(x) D(P \| W(\cdot|x))\\
	&=  \min_{P \in \mcl{P}(\mcl{Y})} \max_{x \in \mcl{X}} \, D(P \| W(\cdot|x)).
\end{align}
The following theorem gives an achievable special-message exponent for the proposed construction. After optimization, this exponent approaches $\widehat{\mcl{J}}$ as the rate approaches capacity.

\begin{theorem}\label{thm:main}
	Fix $Q,\alpha\in\mcl{P}(\mcl{X})$, $R<I(Q,W)$, and a threshold $\Delta<I(Q,W)$. 
	There exists a sequence of codes with a decoder that does not depend on $W$, such that the average error probability of the ordinary messages vanishes and
	\begin{equation}
		E_{\mathrm s}
		\geq
		\widehat{\mcl{J}}_{\alpha}+[\Delta-R]^+.
		\label{eq:main-thm}
	\end{equation}
\end{theorem}

The threshold can be chosen arbitrarily close to $I(QW)$. Choosing a capacity-achieving $Q$,  optimizing over $\alpha$, and taking $\Delta\uparrow C$ gives the following.

\begin{corollary}
For every $R<C$ and every $\epsilon>0$, the exponent $\widehat{\mcl{J}}(R)-\epsilon$ is achievable for the special message, where
\begin{equation}
	\widehat{\mcl{J}}(R)
	\triangleq
	\widehat{\mcl{J}}+C-R.
	\label{eq:Eu}
\end{equation}
\end{corollary}
The optimized expression \eqref{eq:Eu} allows the design parameters to be selected for the channel, but the decoding metric itself never uses the transition probabilities.

We next compare the exponent in~\eqref{eq:Eu} with the natural equal error protection baseline.
 Even without using any UEP structure, one can take an ordinary constant-composition code and simply designate one of its codewords as the special message. 
For a fixed input distribution $Q$, let
\begin{equation}
	 E_{\mathrm{r}}(R,Q)\triangleq  \, \min_{V_{Y|X}} \, D(V\|W|Q)+[I(QV)-R]^+
	 \label{eq:ErQ}
\end{equation}
be the constant-composition random-coding exponent. This exponent is universally attained by MMI decoding~\cite[Ch.~10]{csiszarkorner}.
Therefore, by designating one codeword as special,
the special message can attain $E_{\mathrm{r}}(R)\triangleq\max_Q E_{\mathrm{r}}(R,Q)$ using a universal decoder, while the average error probability of the remaining messages vanishes.
The following proposition shows that the exponent achieved by our construction always dominates this equal error protection baseline.
\begin{proposition}\label{prop-gain}
	For every $R<C$,
\begin{equation}
	 E_{\mathrm{r}}(R)\le C-R\le  \widehat{\mcl{J}}(R).
 \label{eq:random-baseline}
\end{equation}
\end{proposition}
The advantage becomes more apparent as $R\uparrow C$: while $E_{\mathrm r}(R)\downarrow0$, the exponent $\widehat{\mcl{J}}(R)$ converges to $\widehat{\mcl{J}}$. 
Hence, whenever $\widehat{\mcl{J}}>0$, the special message retains a positive exponent even as the transmission rate approaches capacity.
The following proposition shows that this term is positive for every nontrivial channel.
\begin{theorem}\label{prop-positive}
For every DMC, if $C>0$, then
\begin{equation}
	\mcl{J}^{\star}>0
	\qquad\text{and}\qquad
	\widehat{\mcl{J}}>0.
\end{equation}
\end{theorem}
Therefore, as $R\uparrow C$, our construction still gives the special message a strictly positive exponent.
The random-coding exponent is tight at high rates and vanishes at capacity. Under equal error protection, a positive exponent cannot be
guaranteed uniformly over all messages as the rate approaches capacity.
Under UEP, however, the special message can remain exponentially protected while the ordinary messages are still decoded reliably, all with a decoder that does not use the channel transition law.

\begin{proposition}\label{prop-symmetric}
For every symmetric DMC\footnote{A symmetric DMC has rows that are permutations of one another and columns that are permutations of one another.}, $\widehat{\mcl{J}} = \mcl{J}^{\star}.$
\end{proposition}
Thus, for symmetric channels, the proposed universal construction attains the optimal red-alert exponent as the rate approaches capacity.


\begin{remark}
{We assume that $W(y|x)>0$ to ensure that the divergences are finite and the maximum over $\alpha$ is attained.
The bound in Theorem~\ref{thm:main} remains valid when the channel has zero transition probabilities, provided that  terms with $\alpha(x)=0$ are omitted.
In this case, the achievable exponent may be infinite, and the maximum over $\alpha$ should be replaced by a supremum.}
\end{remark}

\section{Numerical Examples}

We now compare the optimal matched exponent $\mcl{J}(R)$, the universal exponent $\widehat{\mcl{J}}(R)$ achieved by the proposed construction, and the  random-coding exponent $E_{\mathrm r}(R)$. 
All rates and exponents are measured in bits.

We first consider a binary symmetric channel with crossover probability $\varepsilon=0.11$. 
Figure~\ref{fig-sym} shows the three exponents over the full range $0\leq R\leq C$.

\begin{figure}[H]
 \centering
\begin{tikzpicture}
\begin{axis}[
  width=7.5cm,
  ylabel={Error Exponent},
  xlabel={$R$ (bits)},
  xmin=0, xmax=0.5,
  ymin=0, ymax=3,
  grid=major, 
  minor tick num=0,
  major grid style={line width=0.3pt, draw=gray!10},  
  legend pos=north east,
  legend style={ 
    draw=black,
    fill=white,
    nodes={inner sep=2.5pt},
  },
]
\addplot [thick, blue, dashed]
  table[x=R_bits,y=J_bits] {bsc-res.dat};
\addlegendentry{$\mathcal{J}(R)$}

\addplot [thick, black, solid]
  table[x=R_bits,y=Es_u_bits] {bsc-res.dat};
\addlegendentry{$\widehat{\mcl{J}}(R)$}

\addplot [thick, red, densely dotted]
  table[x=R_bits,y=Er_bits] {bsc-res.dat};
\addlegendentry{$E_{\mathrm r}(R)$}
\end{axis}
\end{tikzpicture}
\caption{Matched, universal, and random-coding exponents for a $\mathrm{BSC}(0.11)$. 
The universal exponent dominates the random-coding exponent at every rate and remains positive as $R\uparrow C$. 
Since the channel is symmetric, $\widehat{\mcl{J}}(R)$ and $\mcl{J}(R)$ converge to the same value at capacity.}
\label{fig-sym}
\end{figure}
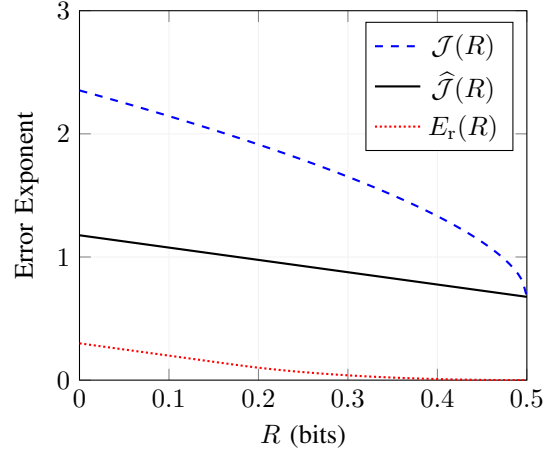
As expected from Proposition~\ref{prop-symmetric}, the matched and universal exponents coincide as the rate approaches capacity. In contrast, the random-coding exponent vanishes. 

We next consider the asymmetric ternary channel
\begin{equation}
W=
\begin{pmatrix}
0.77 & 0.20 & 0.03\\
0.01 & 0.98 & 0.01\\
0.14 & 0.17 & 0.69
\end{pmatrix}.
\label{eq:asymmetric-ternary}
\end{equation}
Figure~\ref{fig-asym} compares the same three exponents over $0\leq R\leq C$.

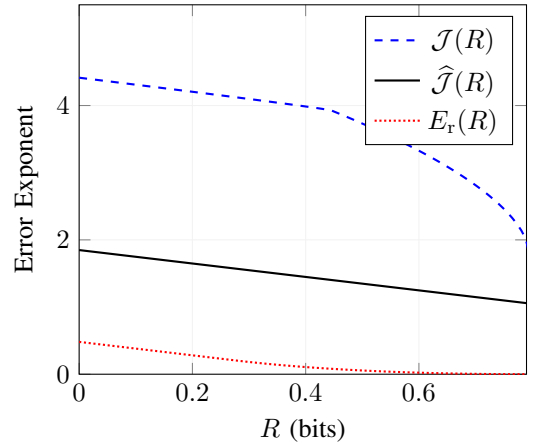
\begin{figure}[H]
 \centering
\begin{tikzpicture}
\begin{axis}[
  width=7.5cm,
  xlabel={$R$ (bits)},
  ylabel={Error Exponent},
  xmin=0, xmax=0.79,
  ymin=0, ymax=5.5,
  grid=major, 
  minor tick num=0,
  major grid style={line width=0.3pt, draw=gray!10},  
  legend pos=north east,
  legend style={ 
    draw=black,
    fill=white,
    nodes={inner sep=2.5pt},
  },
]
\addplot [thick, blue, dashed] table[x=R_bits,y=J_bits] {asym-res.dat};
\addlegendentry{$\mathcal{J}(R)$}
\addplot [thick, black, solid] table[x=R_bits,y=Es_u_bits] {asym-res.dat};
\addlegendentry{$\widehat{\mcl{J}}(R)$}
\addplot [thick, red, densely dotted] table[x=R_bits,y=Er_bits] {asym-res.dat};
\addlegendentry{$E_{\mathrm r}(R)$}

\end{axis}
\end{tikzpicture}
\caption{Matched, universal, and random-coding exponents for the asymmetric ternary channel in~\eqref{eq:asymmetric-ternary}.
The universal exponent again dominates the random-coding exponent at every rate and remains positive as $R\uparrow C$.
Unlike the symmetric case, the gap between $\widehat{\mcl{J}}(R)$ and $\mcl{J}(R)$ remains strictly positive in this limit.}

\label{fig-asym}
\end{figure}

The two examples illustrate the main advantages of the proposed construction. In both cases, $\widehat{\mcl{J}}(R)$ strictly improves on the random-coding exponent and remains positive as the rate approaches capacity. 
For the symmetric channel, the universal and matched exponents coincide as $R\uparrow C$.
Whether a different universal construction can attain the optimal matched exponent $\mcl{J}(R)$ remains open.


\section{Proofs}

\begin{proof}[Proof of Theorem~\ref{thm:main}]
Let $\bs X$ denote an ordinary random codeword, uniformly distributed over $\mcl{A}_n(Q)$. By standard type-class bounds,
\begin{equation}
    |\mcl{A}_n(Q)| = \prod_{x\in\mcl{X}}    |\mcl{T}^{n_x}(Q)|    \doteq    \exp\big(nH(Q)\big).\label{eq-to-ref-1}
\end{equation}
Moreover, since $\mcl{A}_n(Q)\subseteq\mcl{T}^n(Q)$, for every fixed $\bs y\in\mcl{Y}^n$ and $t\ge0$,
\begin{align}
    \mathbb{P}\left[I(\hat P_{\bs X\bs y})\ge t\right] &=
    \frac{    \left|    \left\{     \bs x\in\mcl{A}_n(Q):     I(\hat P_{\bs x\bs y})\ge t    \right\}    \right|}{|\mcl{A}_n(Q)|} \\
    &\le\frac{   \left|    \left\{    \bs x\in\mcl{T}^n(Q):    I(\hat P_{\bs x\bs y})\ge t    \right\}    \right|}{   |\mcl{A}_n(Q)|  }.  
    \label{eq:type-mmi-bound}
\end{align}
Note that
\begin{align}
	  \begin{split}
	  	&\left|    \left\{    \bs x\in\mcl{T}^n(Q):    I(\hat P_{\bs x\bs y})\ge t    \right\}    \right|\\
		&\,\,\, = \sum_{\substack{P_{XY}: P_X=Q, P_Y=\hat{P}_{\bs{y}}\\ I(P)\geq t }} \!\!\!\!\!\! \big|\big\{\bs{x}\in \mcl{T}^n(Q): (\bs{x}, \bs{y}) \in \mcl{T}^n(P)\big\} \big|
	  \end{split}\\
	  &\,\,\, \,\,\dot{\leq} \sum_{\substack{P_{XY}: P_X=Q, P_Y=\hat{P}_{\bs{y}}\\ I(P)\geq t }} \!\!\!\!\!\! \exp\left(nH_P(X|Y)\right)\\
	  &\,\,\, = \sum_{\substack{P_{XY}: P_X=Q, P_Y=\hat{P}_{\bs{y}}\\ I(P)\geq t }} \!\!\!\!\!\! \exp\left(n\left[H(Q) - I(P) \right]\right), \label{eq-to-ref-2}
\end{align}
where the last inequality follows from standard type counting.
Consequently, \eqref{eq-to-ref-1} and \eqref{eq-to-ref-2} imply that, for every fixed $\bs y$,
\begin{align}
    \mathbb{P}\left[I(\hat P_{\bs X\bs y})\ge t\right] \,\dot{\leq}\,\exp(-nt). \label{eq:type-mmi-bound}
\end{align}
We first consider the special message. For every fixed $\bs y$, applying the union bound to \eqref{eq:type-mmi-bound} gives
\begin{align}
    \mathbb{P}\left[
        \max_{m\ge2}\,
        I(\hat P_{\bs X_m\bs y})
        \ge\Delta
    \right] &\,\,\dot\leq\,
    \min\left\{ 1,   \exp\big(n(R-\Delta)\big)    \right\} \\
    \,&\doteq    \exp\left(-n[\Delta-R]^+\right).
    \label{eq:special-mmi-bound}
\end{align}
When the special message is transmitted, the output in block $x$ is generated according to $W(\cdot|x)$. 
For any collection of block output types $\{P_x\}$,
\begin{equation}
 \begin{split}
 	  & \mathbb{P}\big[\hat P_{\bs Y}^x=P_x \text{ for all }x\mid M=1\big]\\
	 &\qquad   \le \exp\left\{-n\sum_x\alpha(x)     D(P_x\|W(\cdot|x))\right\}.
 \end{split}
    \label{eq:proof-block-types}
\end{equation}
Hence, by the method of types,
\begin{align}
    	\begin{split}
		&\mathbb{P}\left[
        \bs Y\notin\mcl{G}_n
        \,|\,
        M=1
    \right] \\
    &\,\,\,\dot\leq
    \exp\Bigg(
    -n \!\!\!\!\!\!\!\!\!
    \min_{\substack{
        \{P_x\}_{x\in\mcl{X}}:\\
       \underset{x\neq \bar{x}}{\max}
        \|P_x-P_{\bar{x}}\|_\infty
        \le\delta_n
    }}
    \!\!\!\sum_{x}
    \alpha(x)
    D\big(P_x\|W(\cdot|x)\big)
    \Bigg).
	\end{split}
        \label{eq:homogeneity-bound}
\end{align}
Since $\delta_n\downarrow0$, the minimization in~\eqref{eq:homogeneity-bound} converges to
\begin{equation}
    \min_{P\in\mcl{P}(\mcl{Y})}
    \sum_{x\in\mcl{X}}
    \alpha(x)
    D\big(P\|W(\cdot|x)\big)
    =
    \widehat{\mcl{J}}_\alpha.
\end{equation}
Therefore,
\begin{equation}
    \mathbb{P}\left[
        \bs Y\notin\mcl{G}_n\,|\,   M=1
    \right]
    \,\dot\leq\,
    \exp\left(-n\widehat{\mcl{J}}_\alpha\right).
    \label{eq:homogeneity-exponent}
\end{equation}
An error under the special message can occur only if $\bs Y\notin\mcl{G}_n$ and at least one ordinary codeword has empirical mutual information at least $\Delta$. Combining~\eqref{eq:special-mmi-bound} and~\eqref{eq:homogeneity-exponent}, and averaging over the random codebook, gives
\begin{equation}
    \!\!\mathbb{E}\!\left[
        \mathbb{P}\big[
            \hat M\neq1
            \,\big|\,
            M=1
        \big]
    \right]
    \dot\leq
    \exp\!\left(-n\left[\widehat{\mcl{J}}_\alpha + [\Delta-R]^+ \right]
    \right).
    \label{eq:special-average-bound}
\end{equation}

We next consider an ordinary message $m\ge2$. Since every block of $\bs X_m$ has composition $Q$, the output type in every block converges to $(QW)_Y$. 
Hoeffding's inequality and the union bound give
\begin{equation}
    \mathbb{P}\big[\bs Y\in\mcl{G}_n\mid M=m\big]
    \le
    2|\mcl{Y}|\sum_x e^{-n \frac{\alpha(x)\delta_n^2}{2}}.
    \label{eq:proof-homogeneity-ordinary}
\end{equation}

Since $n\delta_n^2\to\infty$, it follows that
\begin{equation}
    \mathbb{P}\left[
        \bs Y\in\mcl{G}_n
        \,\middle|\,
        M=m
    \right]
    \longrightarrow0.
\end{equation}

Moreover, by the law of large numbers, the joint type $\hat P_{\bs X_m\bs Y}$ converges in probability to $Q\times W$. By continuity of mutual information,
\begin{equation}
    I(\hat P_{\bs X_m\bs Y})
    \longrightarrow
    I(QW)
\end{equation}
in probability. Since $\Delta<I(QW)$,
\begin{equation}
    \mathbb{P}\left[
        I(\hat P_{\bs X_m\bs Y})<\Delta
        \,\middle|\,
        M=m
    \right]
    \longrightarrow0.
\end{equation}
It remains to consider errors among the ordinary messages. 
Since $R<I(QW)$, choose $\epsilon>0$ such that
\begin{equation}
    R+\epsilon<I(QW).
\end{equation}
Consider an ordinary message $m\ge2$, and define the event
\begin{equation}
    \mcl{E}_m
    \triangleq
    \left\{
    I(\hat P_{\bs X_m\bs Y})>R+\epsilon
    \right\}.
\end{equation}
Suppose that $\mcl{E}_m$ occurs, and fix the realizations
$(\bs X_m,\bs Y)=(\bs x_m,\bs y)$. For every competing ordinary
message $\bar m\neq m$, the codeword $\bs X_{\bar m}$ is independent of
$(\bs X_m,\bs Y)$. Hence, by the previous bound in \eqref{eq:type-mmi-bound},
\begin{align}
    \begin{split}
    	&\mathbb{P}\left[
    I(\hat P_{\bs X_{\bar m}\bs y})
    \ge
    I(\hat P_{\bs x_m\bs y})
    \,|\,
    \bs X_m=\bs x_m,\bs Y=\bs y
    \right]
    \\
    &\,\,\qquad\qquad\qquad\dot{\leq}
    \exp\left(
    -nI(\hat P_{\bs x_m\bs y})
    \right)
    \end{split}    \\
    &\qquad\qquad\qquad\le
    \exp\left(-n(R+\epsilon)\right).
\end{align}
Applying the union bound over the ordinary competing messages therefore gives
\begin{align}
	\begin{split}
		 &\mathbb{P}\left[
    \exists\,\bar m\neq m:
    I(\hat P_{\bs X_{\bar m}\bs Y})
    \ge
    I(\hat P_{\bs X_m\bs Y})
    \,|\,
    \mcl{E}_m,M=m
    \right]\\
   &\qquad\qquad\qquad \dot{\leq}\,
    \exp(-n\epsilon).
	\end{split}
\end{align}
Consequently,
\begin{align}
    &\mathbb{P}\left[
    \exists\,\bar m\neq m:
    I(\hat P_{\bs X_{\bar m}\bs Y})
    \ge
    I(\hat P_{\bs X_m\bs Y})
    \,\middle|\,
    M=m
    \right]
    \nonumber\\
    &\qquad\le
    \mathbb{P}\left[\mcl{E}_m^c\mid M=m\right]
    +\exp(-n\epsilon+o(n)).
    \label{eq:ordinary-mmi-error}
\end{align}
By the convergence of $I(\hat P_{\bs X_m\bs Y})$ to $I(QW)$ and the choice
$R+\epsilon<I(QW)$,
\begin{equation}
    \!\!\!\!\mathbb{P}\left[\mcl{E}_m^c| M=m\right]
    =
    \mathbb{P}\!\left[
    I(\hat P_{\bs X_m\bs Y})\!\le\! R+\epsilon
    \,|\,
    M=m
    \right]
    \to 0.
\end{equation}
It follows from~\eqref{eq:ordinary-mmi-error} that the ensemble average probability of an MMI decoding error among the ordinary messages vanishes.

Let $p_{\mathrm o}(\mcl C)$ and $p_{\mathrm s}(\mcl C)$ denote the
ordinary and special-message error probabilities of a  codebook $\mcl C$, respectively. By Markov's inequality,
\begin{align}
    \mathbb{P}\left[
    p_{\mathrm o}(\msf C)>
    3\mathbb{E}[p_{\mathrm o}(\msf C)]
    \right]
    &\le \frac{1}{3},\\
    \mathbb{P}\left[
    p_{\mathrm s}(\msf C)>
    3\mathbb{E}[p_{\mathrm s}(\msf C)]
    \right]
    &\le \frac{1}{3}.
\end{align}
Therefore, by the union bound, there exists a deterministic codebook
for which both
\begin{equation}
    p_{\mathrm o}(\mcl C)
    \le
    3\mathbb{E}[p_{\mathrm o}(\msf C)]
\end{equation}
and
\begin{equation}
    p_{\mathrm s}(\mcl C)
    \le
    3\mathbb{E}[p_{\mathrm s}(\msf C)]
\end{equation}
hold simultaneously. Since the first expectation vanishes and the constant does not affect the exponent of the second,
there exists a deterministic codebook such that the average error probability of the ordinary messages vanishes and
\begin{equation}
    E_{\mathrm s}
    \ge
    \widehat{\mcl J}_\alpha+[\Delta-R]^+.
\end{equation}

\end{proof}

\begin{proof}[Proof of Proposition~\ref{prop-gain}]

Choosing $V=W$ in~\eqref{eq:ErQ} gives
\begin{equation}
    E_{\mathrm r}(R,Q)
    \le [I(QW)-R]^+
    \le C-R,
\end{equation}
where the last inequality follows from $I(QW)\le C$ and $R<C$.
Maximizing over $Q$ yields $E_{\mathrm r}(R)\le C-R$.
The second inequality follows from
$\widehat{\mcl{J}}(R)=\widehat{\mcl{J}}+C-R$ and
$\widehat{\mcl{J}}\geq0$.
\end{proof}

\begin{proof}[Proof of Theorem~\ref{prop-positive}]
Since $C>0$, the channel contains at least two distinct rows.
Suppose first that $\mcl{J}^{\star}=0$. Since relative entropy is nonnegative,
\begin{equation}
    D\big(P_Y^\star\|W(\cdot|x)\big)=0
\end{equation}
for every $x\in\mcl{X}$. Hence,
\begin{equation}
    W(\cdot|x)=P_Y^\star
\end{equation}
for every $x$, so all rows of the channel are identical. This implies $C=0$, which is a contradiction. Therefore,
\begin{equation}
    \mcl{J}^{\star}>0.
\end{equation}
Similarly, suppose that $\widehat{\mcl{J}}=0$. Since the minimum in
\begin{equation}
    \widehat{\mcl{J}}
    =
    \min_{P\in\mcl{P}(\mcl{Y})}
    \max_{x\in\mcl{X}}
    D\big(P\|W(\cdot|x)\big)
\end{equation}
is attained, there exists $P^\star\in\mcl{P}(\mcl{Y})$ such that
\begin{equation}
    D\big(P^\star\|W(\cdot|x)\big)=0
\end{equation}
for every $x\in\mcl{X}$. Thus, $W(\cdot|x)=P^\star$ for every $x$, again implying $C=0$. Hence, $\widehat{\mcl{J}}>0$.

\end{proof}

\begin{proof}[Proof of Proposition~\ref{prop-symmetric}]
Let $P_{\mathrm u}$ denote the uniform distribution on $\mcl{Y}$.
By symmetry, the uniform input distribution induces $P_{\mathrm u}$,
and
\begin{equation}
    \mcl{J}^\star
    =
    D(P_{\mathrm u}\|W(\cdot|x))
    \qquad\text{for every }x\in\mcl{X}.
\end{equation}
Since $\mcl{J}^\star$ is the optimal red-alert exponent,
$\widehat{\mcl{J}}\leq\mcl{J}^\star$.
Conversely, for every $P\in\mcl{P}(\mcl{Y})$, symmetry gives
\begin{align}
    \max_{x\in\mcl{X}}D(P\|W(\cdot|x))
    &\geq
    \frac{1}{|\mcl{X}|}
    \sum_{x\in\mcl{X}}D(P\|W(\cdot|x))\\
    &=
    D(P\|P_{\mathrm u})+\mcl{J}^\star\\
    &\geq
    \mcl{J}^\star.
\end{align}
Minimizing over $P$ yields
$\widehat{\mcl{J}}\geq\mcl{J}^\star$, and hence
\begin{equation}
    \widehat{\mcl{J}}=\mcl{J}^\star.
\end{equation}
\end{proof}

\section*{acknowledgement}
The proposed code construction was developed by the authors, who also carried out all mathematical analysis and proofs in this work.
OpenAI GPT-6 Astra was used to assist in checking the proofs and mathematical statements and refining the presentation of technical arguments.
It also assisted with running the simulations and obtaining the numerical results reported in Section~IV.
All AI-assisted outputs were independently verified by the authors.

{\footnotesize
\bibliographystyle{IEEEtran}
\bibliography{refs}
}

\end{document}